\documentclass[fontsize=9pt]{llncs}
  \usepackage{microtype}
  \usepackage{amsmath,amsfonts}
  \usepackage{tikz}
  \usepackage[T1]{fontenc}
  \usepackage[utf8]{inputenc}
  \usepackage{comment}
  \usepackage{forest}
  \usepackage{xspace}
  \usepackage{enumerate}
  \usepackage{listings}
  \usepackage{multirow}
  \usepackage{todonotes}
  \usepackage{thmtools}
  \usepackage{thm-restate}
  \usepackage[noend, boxed, noline]{algorithm2e}
  \usepackage[capitalise]{cleveref}
  \usepackage{graphics}
  \usepackage{tabularx}
  \usepackage{amssymb}
  \usepackage{epsfig}
  \usepackage{vaucanson-g}
  \usepackage[whileod]{chl-alg}
  \usepackage[caption=true,font=footnotesize]{subfig}
  \usepackage{verbatim}
  \usetikzlibrary{trees, arrows, shapes}
\tikzset{
  treenode/.style = {align=center, inner sep=2pt, text centered,
    font=\sffamily},
  arn_r/.style = {treenode, circle, black, font=\sffamily\bfseries, draw=black,
    text width=1.5em},
    arn_t/.style = {treenode, circle, black, thick, double, font=\sffamily\bfseries, draw=black,
    text width=1.5em},
  every edge/.append style={anchor=south,auto=falseanchor=south,auto=false,font=3.5 em},
}
\newcommand{\std}{\textsl{std}}

  \newtheorem{lem}[theorem]{Lemma}
  \renewenvironment{lemma}{\begin{lem}}{\end{lem}}
  \crefname{lem}{Lemma}{Lemmas}

\def\dd{\mathinner{.\,.}}
\newcommand{\cO}{\mathcal{O}}

 \newcommand{\defproblem}[3]{
  \vspace{2mm}
\noindent\fbox{
  \begin{minipage}{0.96\textwidth}
  #1\\
  {\bf{Input:}} #2  \\
  {\bf{Output:}} #3
  \end{minipage}
  }
  \vspace{2mm}
}

\begin{document}
\title{Optimal Computation of Avoided Words}

\author{
Yannis Almirantis \inst{1}
\and
Panagiotis Charalampopoulos \inst{2}
\and
Jia Gao \inst{2}
\and
Costas~S.~Iliopoulos \inst{2}
\and
Manal Mohamed \inst{2}
\and
Solon~P.~Pissis \inst{2}
\and \\
Dimitris Polychronopoulos  \inst{3}
}

\institute{
    National Center for Scientific Research Demokritos, Athens, Greece 
\email{yalmir@bio.demokritos.gr}
\and 
    Department of Informatics, King's College London, UK\\
    \email{\{panagiotis.charalampopoulos,jia.gao,costas.iliopoulos,\\manal.mohamed,solon.pissis\}@kcl.ac.uk}\\[1ex]
\and 
MRC Clinical Sciences Centre, Imperial College London, UK\\
 \email{d.polychronopoulos@csc.mrc.ac.uk}  
   }

\date{}

\maketitle 
   
\begin{abstract}
The deviation  of the observed frequency of a word $w$ from its expected frequency in a given sequence $x$ is used to determine whether or not the word is {\em avoided}. This concept is particularly useful in DNA linguistic analysis. The value of the standard deviation of $w$, denoted by $\std(w)$, effectively characterises the extent of a word by its edge contrast in the context in which it occurs. A word $w$ of length $k>2$ is a $\rho$-avoided word in $x$ if $\std(w) \leq \rho$, for a given threshold $\rho < 0$. Notice that such a word may be completely {\em absent} from $x$. Hence computing all such words na\"{\i}vely can be a very time-consuming procedure, in particular for large $k$. In this article, we propose an $\cO(n)$-time and $\cO(n)$-space algorithm to compute all $\rho$-avoided words of length $k$ in a given sequence $x$ of length $n$ over a fixed-sized alphabet. We also present a time-optimal $\cO(\sigma n)$-time and $\cO(\sigma n)$-space algorithm to compute all $\rho$-avoided words (of any length) in a sequence of length $n$ over an alphabet of size $\sigma$. Furthermore, we provide a tight asymptotic upper bound for the number of $\rho$-avoided words and the expected length of the longest one. We make available an open-source implementation of our algorithm. Experimental results, using both real and synthetic data, show the efficiency of our implementation. 
\end{abstract}

\section{Introduction}
%

The one-to-one mapping of a DNA molecule to a sequence of letters suggests that DNA analysis can be modelled within the framework of formal language theory~\cite{10.2307/29774782}. For example, a region within a DNA sequence can be considered as a ``word'' on a fixed-sized alphabet in which some of its natural aspects can be described by means of certain types of automata or grammars. However, a linguistic analysis  of the DNA  needs to take into account many distinctive physical and biological characteristics of such sequences:  DNA contains coding regions that encode for polypeptide chains associated with biological functions; and non-coding regions, most of which are not linked to any particular function. Both appear to have many statistical features in common with natural languages~\cite{1994PhRvL..73.3169M}. 

A computational tool oriented towards the systematic search for avoided words is particularly useful for {\em in silico} genomic research analyses. The search for {\em absent words} is already undertaken in the recent past and several results exist~\cite{nullrly}. However, words which may be present in a genome or in genomic sequences of a specific role (e.g.,~protein coding segments, regulatory elements, conserved non-coding elements etc) but they are strongly underrepresented---as we can estimate on the basis of the frequency of occurrence of their longest proper factors---may be of particular importance. They can be words of nucleotides which are hardly tolerated because they negatively influence the stability of the chromatin or, more generally, the functional genomic conformation; they can represent targets of restriction endonucleases which may be found in bacterial and viral genomes; or, more generally, they may be short genomic regions whose presence in wide parts of the genome are not tolerated for less known reasons. The understanding of such avoidances is becoming an interesting line of research (for recent studies, see~\cite{DBLP:conf/spire/BelazzouguiC15,Rusinov2015}). 

On the other hand, short words of nucleotides may be systematically avoided in large genomic regions or whole genomes for entirely different reasons: just because they play important signaling roles which restrict their appearance only in specific positions: consensus sequences for the initiation of gene transcription and of DNA replication are well-known such oligonucleotides. Other such cases may be insulators, sequences anchoring the chromatin on the nuclear envelope like lamina-associated domains, short sequences like dinucleotide repeat motifs with enhancer activity, and several other cases. Again, we cannot exclude that this area of research could lead to the identification of short sequences of regulatory activities still unknown.

Brendel et al. in \cite{brendel1986linguistics} initiated research into the linguistics of nucleotide sequences that focuses on  the concept of words in continuous languages---languages devoid of blanks---and introduced an operational definition of words. The authors suggested a method to measure, for each possible word $w$ of length $k$, the deviation of its observed frequency from the expected frequency in a given sequence.  The values of the standard deviation, denoted by $\std(w)$, were then used to identify words that are avoided among all possible words of length $k$. The typical length of avoided (or of overabundant) words of the nucleotide language was found to range from 3 to 5 (tri- to pentamers).
The statistical significance of the avoided words was shown to reflect their biological importance. This work, however, was based on the very limited sequence data available at the time: only DNA sequences from two viral and one bacterial genomes were considered. Also note that $k$ might change when considering eukaryotic genomes, the complex dynamics and function of which might impose a more demanding analysis.   

\noindent \textbf{Our contribution.} The computational problem can be described as follows. Given a sequence $x$ of length $n$, an integer $k$, and a real number $\rho < 0$, compute the set of $\rho$-avoided words of length $k$, i.e. all words $w$ of length $k$ for which $\std(w) \leq \rho$. We call this set the {\em $\rho$-avoided words} of length $k$ in $x$. Brendel et al.~did not provide an efficient solution for this computation~\cite{brendel1986linguistics}. Notice that such a word may be completely absent from $x$. Hence the set of $\rho$-avoided words can be na\"{\i}vely computed by considering all possible $\sigma^k$ words, where $\sigma$ is the size of the alphabet. Here we present an $\cO(n)$-time and $\cO(n)$-space algorithm for computing all $\rho$-avoided words of length $k$ in a sequence $x$ of length $n$  over a fixed-sized alphabet. We also present a time-optimal $\cO(\sigma n)$-time and $\cO(\sigma n)$-space algorithm to compute all $\rho$-avoided words (of any length) over an integer alphabet of size $\sigma$. Furthermore, we provide a tight asymptotic upper bound for the number of $\rho$-avoided words and the expected length of the longest one. We make available an open-source implementation of our algorithm. Experimental results, using both real and synthetic data, show its efficiency and applicability. Specifically, using our method we confirm that restriction endonucleases which target self-complementary sites are not found in eukaryotic sequences~\cite{Rusinov2015}.

\section{Terminology and Technical Background} 
\subsection{Definitions and Notation}
We begin with basic definitions and notation generally following~\cite{CHL07}. Let $x=x[0]x[1]\dd x[n-1]$ be a \textit{word} of \textit{length} $n=|x|$ over a finite ordered \textit{alphabet} $\Sigma$ of size $\sigma = |\Sigma|=\cO(1)$. For two positions $i$ and $j$ on $x$, we denote by $x[i\dd j]=x[i]\dd x[j]$ the \textit{factor} (sometimes called \textit{subword}) of $x$ that starts at position $i$ and ends at position $j$ (it is empty if $j < i$), and by $\varepsilon$ the \textit{empty word}, word of length 0.  We recall that a prefix of $x$ is a factor that starts at position 0 
($x[0\dd j]$) and a suffix is a factor that ends at position $n-1$ 
($x[i\dd n-1]$), and that a factor of $x$ is a \textit{proper} factor if 
it is not $x$ itself. A factor of $x$ that is neither a prefix nor a suffix of $x$ is called an $\textit{infix}$ of $x$.

Let $w=w[0]w[1]\dd w[m-1]$ be a word, $0<m\leq n$. 
  We say that there exists an \textit{occurrence} of $w$ in $x$, or, more 
simply, that $w$ \textit{occurs in} $x$, when $w$ is a factor of $x$.
  Every occurrence of $w$ can be characterised by a starting position in $x$. 
  Thus we say that $w$ occurs at the \textit{starting position} $i$ in $x$ 
when $w=x[i \dd i + m - 1]$. Further let $f(w)$ denote the \textit{observed frequency}, that is, the number of occurrences of $w$ in word $x$. If $f(w) = 0$ for some word $w$, then $w$ is called \textit{absent}, otherwise, $w$ is called \textit{occurring}.

By $f(w_p)$, $f(w_s)$, and $f(w_i)$ we denote the observed frequency of the longest proper prefix $w_p$, suffix $w_s$, and infix $w_i$ of $w$ in $x$, respectively. We can now define the \textit{expected frequency} of word $w$ in $x$ as in Brendel et al.~\cite{brendel1986linguistics}:  

\begin{equation} \label{eq:1}
E(w) =  \frac {f(w_p) \times f(w_s)}{f(w_i)}, \text{ if~ } f(w_i) >0;  \text{~else~}  E(w) = 0.
\end{equation}
The above definition can be explained intuitively as follows. Suppose we are given $f(w_p)$, $f(w_s)$, and $f(w_i)$. Given an occurrence of $w_i$ in $x$, the probability of it being preceded by $w[0]$ is $\frac {f(w_p)}{f(w_i)}$ as $w[0]$ precedes exactly $f(w_p)$ of the $f(w_i)$ occurrences of $w_i$. Similarly,  this occurrence of $w_i$ is also an occurrence of $w_s$ with probability $\frac {f(w_s)}{f(w_i)}$. Although these two events are not always independent, the product $\frac {f(w_p)}{f(w_i)} \times \frac {f(w_s)}{f(w_i)}$ gives a good approximation of the probability that an occurrence of $w_i$ at position $j$ implies an occurrence of $w$ at position $j-1$. It can be seen then that by multiplying this product by the number of occurrences of $w_i$ we get the above formula for the expected frequency of $w$.

\noindent Moreover, to measure the deviation of the observed frequency of a word $w$ from its expected frequency in $x$, we define the {\em standard deviation} ($\chi^2$ test) of $w$ as:

\begin{equation} \label{eq:2}
 \std(w) = \frac {f(w)-E(w)}{max\{ \sqrt {E(w)}, 1\}}.
\end{equation}

\noindent For more details on the {\em biological} justification of these definitions see ~\cite{brendel1986linguistics}.

Using the above definitions and a given threshold, we are in a position to classify a word $w$ as either {\em avoided} or {\em common} in $x$. In particular, for a given threshold $\rho < 0$, a word $w$ is called $\rho$-{\em avoided} if $\std(w) \leq \rho$. In this article, we consider the following 
computational problem. 

{\defproblem{\textsc{AvoidedWordsComputation}}{A word $x$ of length $n$, an integer $k>2$, and a real number $\rho < 0$}{All $\rho$-avoided words of length $k$ in $x$}}


\subsection{Suffix Trees}

In our algorithm, suffix trees are used  extensively as computational tools. For a general introduction to suffix trees, see~\cite{CHL07}.

The \textit{suffix tree} $\mathcal{T}(x)$ of a non-empty word $x$ of length $n$ is a compact trie representing all suffixes of $x$, the nodes of the trie which become nodes of the suffix
tree are called {\em explicit} nodes, while the other nodes are called {\em implicit}. Each edge
of the suffix tree can be viewed as an upward maximal path of implicit nodes starting with an explicit node. Moreover, each node belongs to a unique path of that kind. Then, each node of the trie can be represented in the suffix tree by the edge it belongs to and an index within the corresponding path.

We use  $\mathcal{L}(v)$ to  denote the \textit{path-label} of a node $v$, i.e., the concatenation of the edge labels along the path from the root to $v$. We say that $v$ is  path-labelled  $\mathcal{L}(v)$. Additionally, $\mathcal{D}(v)= |\mathcal{L}(v)|$ is used to denote  the \textit{word-depth} of node $v$. Node  $v$ is  a \textit{terminal} node, if and only if, $\mathcal{L}(v) = x[i \dd n]$, $0 \leq i < n$; here $v$ is also labelled with index $i$.    It should be clear that  each  occurring word $w$ in  $x$ is uniquely represented by either an explicit or implicit node of $\mathcal{T}(x)$. The \textit{suffix-link} of a node $v$  with path-label $\mathcal{L}(v)= ay$ is a pointer to the node path-labelled $y$, where  $a \in \Sigma$ is a single letter and $y \in  \Sigma^*$ is a word. The suffix-link of $v$ exists if $v$ is a non-root internal node of $\mathcal{T}(x)$. 

In any standard implementation of the suffix tree, we assume that each node of the suffix tree is able to access its parent. Note that, once $\mathcal{T}(x)$ is constructed, it can be traversed to compute the word-depth $\mathcal{D}(v)$ for  each node $v$. The tree is traversed in a depth-first manner, for each  node $v$. Let $u$ be the parent of $v$. Then the word-depth $\mathcal{D}(v)$ is computed by adding  $\mathcal{D}(u)$ to the length of the label of edge $(u,v)$. If $v$ is the root then $\mathcal{D}(v) = 0$. Additionally, a depth-first traversal of $\mathcal{T}(x)$ allows us to count,  for  each  node $v$, the number of terminal nodes in the subtree rooted at $v$, denoted by $\mathcal{C}(v)$,   as follows. When internal node $v$  is visited, $\mathcal{C}(v)$ is computed by adding up $\mathcal{C}(u)$ of all the nodes $u$, such that $u$ is a child of $v$, and then $\mathcal{C}(v)$ is incremented by 1 if $v$ itself is a terminal node. If a node $v$ is a leaf  then $\mathcal{C}(v) = 1$.

\section{Useful Properties}
In this section, we provide some useful insights of computational nature which were not
considered by Brendel et al.~\cite{brendel1986linguistics}. By the definition of $\rho$-avoided words it follows that a word $w$ may be $\rho$-avoided even if it is absent from $x$. In other words, $\std(w) \leq \rho$ may hold for either  $f(w) > 0$ (occurring) or $f(w) = 0$ (absent).

This means that a na\"{\i}ve computation should consider {\em all} possible $\sigma^k$ words. Then for each possible word $w$, the value of $\std(w)$ can be computed via pattern matching on the suffix tree. In particular we can search for the occurrences of $w$, $w_p$, $w_s$, and $w_i$ in time 
$\cO(k)$~\cite{CHL07}. In order to avoid this inefficient computation, we exploit the following crucial lemmas.
\begin{definition}[\cite{Barton2014}]~\label{def:maws}
An absent word $w$ of $x$ is {\em minimal} if and only if all its proper factors occur in $x$. 
\end{definition}
\begin{lemma}\label{lem1}
Any absent $\rho$-avoided word $w$ in $x$ is a minimal absent word of $x$.
\end{lemma}
\begin{proof}
For $w$ to be a $\rho$-avoided word it must hold that
$$\std(w) = \frac {f(w)-E(w)}{max\{ \sqrt {E(w)}, 1\}}\leq \rho < 0.$$
This implies that $f(w)-E(w)<0$, which in turn implies that $E(w)>0$ since $f(w) = 0$.
From $E(w) =  \frac {f(w_p) \times f(w_s)}{f(w_i)}>0$, we conclude that $f(w_p)>0$ and $f(w_s)>0$ must hold. Since $f(w) = 0$, $f(w_p)>0$, and $f(w_s)>0$, $w$ is a minimal absent word of $x$:
all proper factors of $w$ occur in $x$. \qed
\end{proof}
\begin{lemma}\label{lem:occ}
Let $w$ be a word occurring in $x$ and $\mathcal{T}(x)$ be the suffix tree of $x$. 
Then, if $w_p$  is a  path-label of an implicit node of $\mathcal{T}(x)$, $std(w) \geq 0$.
\end{lemma}
\begin{proof}
Since $w$ occurs in $x$ it holds that $f(w_i) \geq f(w_s)$ and, hence, by the definition of $E(w)$, 
$f(w_p) \geq E(w)$. Furthermore, by the definition of the suffix tree, since $w$ occurs in $x$ and $w_p$ is a path-label of an implicit node then $f(w_p) = f(w)$. It thus follows that 
$f(w) - E(w) = f(w_p) - E(w) \geq 0$, and since $\max\{1,\sqrt{E(w)}\} > 0$, the claim holds. \qed
\end{proof}
\begin{lemma}\label{lem2}
The total number of $\rho$-avoided words of length $k>2$ in a word $x$ of length $n$ over an alphabet of size $\sigma$ is bounded from above by $\cO(\sigma n)$; in particular, this number is no more than $(\sigma + 1) n - k + 1$.
\end{lemma}
\begin{proof}
By Lemma~\ref{lem1}, every $\rho$-avoided word is either occurring or a minimal absent word.
It is known that the total number of minimal absent words in $x$ is smaller than or equal to $\sigma n$~\cite{Mignosi02}. Clearly, the occurring $\rho$-avoided words in $x$ are at most $n - k + 1$. Therefore the lemma holds. \qed 
\end{proof}
\begin{example}\label{exa:suffixTree}
Consider the word $x=\texttt{AGCGCGACGTCTGTGT}$. Fig.~\ref{Fig:suffixTree} represents the suffix tree $\mathcal{T}(x)$. Note that word $\texttt{GCG}$ is represented by the explicit internal node $v$; whereas word $\texttt{TCT}$ is represented by the implicit node along the edge connecting the node labelled 15 and the node labelled 9. Consider node $v$ in $\mathcal{T}(x)$; we have that $\mathcal{L}(v) = \texttt{GCG}$, $\mathcal{D}(v) = 3$, and $\mathcal{C}(v)=2$.  

\begin{figure}[h]\begin{center}
\label{suffixtree}
\resizebox{!}{0.50\totalheight}
{
\begin{tikzpicture}[->,>=stealth',level/.style={level 1/.style={sibling distance=7cm},
  level 2/.style={sibling distance=2.8cm},
  level 3/.style={sibling distance=1.5cm},level distance = 3cm}] 
\node [arn_r] (tree){}
  child{node [arn_r](L1-C1){}
  child{node [arn_t](L2-C1){6} edge from parent node[left]{$\texttt{CGTC}\ldots$}}  
  child{node [arn_t](L2-C5){0} edge from parent node[right]{$\texttt{GCGC}\ldots$}}
  edge from parent node[left]{$\texttt{A}$}
}	
child{node [arn_r](L1-C2){}
  child{node [arn_r](L3-C6){} 
  child{node [arn_t]{4} edge from parent node[left]{$\texttt{ACGT}\ldots$}
  }
  child{node [arn_t]{2} edge from parent node[above]{$\texttt{CGAC}\ldots$}
  }
  child{node [arn_t]{7} edge from parent node[right]{$\texttt{TCTC}\ldots$}
  }
  edge from parent node[left]{$\texttt{G}$}
  }
  child{node [arn_t]{10} edge from parent node[right]{$\texttt{TGTGT}$}
  }
  edge from parent node[left]{$\texttt{C}$}
}
child{node [arn_r](L1-C3){}
child{node [arn_t]{5} edge from parent node[left]{$\texttt{ACGT}\ldots$}}
child{node [arn_r](vnode){$v$}
 child{node [arn_t]{3} edge from parent node[left]{$\texttt{ACGT}\ldots$}}
child{node [arn_t]{1} edge from parent node[above]{$\texttt{CGAC}\ldots$}}
edge from parent node[ above] {$\texttt{CG}$}}  
child{node [arn_t](L3-C2){14}
child{node [arn_t]{8} edge from parent node[above]{$\texttt{CTGTGT}$}}   
child{node [arn_t]{12} edge from parent node[right]{$\texttt{GT}$}}
edge from parent node[right]{$\texttt{T}$}}
edge from parent node[right]{$\texttt{G}$}}
child{node [arn_t](L1-C4){15}
child{node [arn_t] {9} edge from parent node[left]{$\texttt{CTGTGT}$}}
child{node [arn_t](L3-C4){13}
child{node [arn_t]{11} edge from parent node[right]{$\texttt{GT}$}}
edge from parent node[right]{$\texttt{GT}$}}
edge from parent node[right]{$\texttt{T}$}
};
\draw[semithick,dashed,->] (L1-C1) to [bend right=-20] (tree); 
\draw[semithick,dashed,->] (L1-C2) to [bend right=-20] (tree); 
\draw[semithick,dashed,->] (L1-C3) to [bend right=20] (tree); 
\draw[semithick,dashed,->] (L1-C4) to [bend right=20] (tree); 
\draw[semithick,dashed,->] (L3-C2) to [bend right=0] (L1-C4);
\draw[semithick,dashed,->] (L3-C4) to [bend right=-20] (L3-C2);
\draw[semithick,dashed,->] (L3-C6) to [bend right=0] (L1-C3);
\draw[semithick,dashed,->] (vnode) to [bend right=-20] (L3-C6);
\end{tikzpicture}
}
\end{center}
\caption{The suffix tree $\mathcal{T}(x)$ for $x= \texttt{AGCGCGACGTCTGTGT}$. Double-lined nodes represent terminal nodes labelled with the associated indices. The  suffix-links for non-root internal nodes are dashed.}
\label{Fig:suffixTree}
\end{figure}
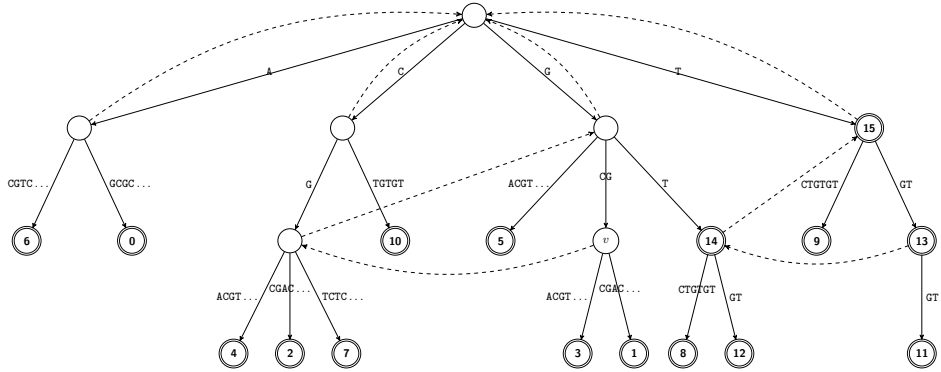
\end{example}

\begin{example}
Consider the word from Example~\ref{exa:suffixTree}, $k=3$, and $\rho=-0.4 $. 
\begin{itemize}
\item word  $w_1= \texttt{CGT}$, at position 7 of $x$, is an {\em occurring} $\rho$-avoided word: 
$$ E(w_1) = 3\times 3/6 = 1.5,\text{ } \std(w_1) =(1-1.5)/\sqrt{1.5} = -0.408248.$$
\item word $w_2 = \texttt{AGT}$ is an {\em absent} $\rho$-avoided word:  
$$ E(w_2) = 1\times 3/6 = 0.5,\text{ } \std(w_2) =(0- 0.5)/1 = -0.5.$$
\end{itemize}
\end{example}

\section{Avoided  Words Algorithm}
In this section, we present Algorithm  \Algo{AvoidedWords} for computing
all $\rho$-avoided words of length $k$ in a given word $x$. The algorithm builds the suffix tree $\mathcal{T}(x)$ for word $x$, and then prepares  $\mathcal{T}(x)$ to allow constant-time  observed frequency queries. This is mainly achieved by counting the terminal nodes in the subtree rooted at node $v$ for every node $v$ of  $\mathcal{T}(x)$.  Additionally during  this preprocessing, the algorithm  computes the word-depth of $v$ for every node $v$ of  $\mathcal{T}(x)$. By Lemma~\ref{lem1}, $\rho$-avoided words are classified as either occurring or minimal absent, therefore Algorithm \Algo{AvoidedWords} calls Routines \Algo{AbsentAvoidedWords} and \Algo{OccurringAvoidedWords} to compute both classes of $\rho$-avoided words in $x$. The outline of Algorithm \Algo{AvoidedWords} is as follows.

\begin{algo}{AvoidedWords}{$$x$, $k$, $\rho$$}
\SET {\mathcal{T}(x)}{\Call{SuffixTree}{x}}
\DOFOR{\mbox{each node  $v \in \mathcal{T}(x)$} }
\SET{\mathcal{D}(v)}{\mbox{word-depth of }v}
\SET{\mathcal{C}(v)}{\mbox{number of terminal nodes in the subtree rooted at }v}
\OD
\CALL{AbsentAvoidedWords}{x,k,\rho}
\CALL{OccurringAvoidedWords}{x,k,\rho}
\end{algo}

\subsection{Computing Absent Avoided  Words}

In Lemma~\ref{lem1}, we showed that each absent $\rho$-avoided word is a minimal absent word. Thus, Routine \Algo{AbsentAvoidedWords} starts by computing all minimal absent words in $x$; this can be done in time and space $\cO(n)$ for a fixed-sized alphabet or in time $\cO(\sigma n)$ for large alphabets~\cite{Barton2014}. Let $< (i,j), \alpha  >$ be a tuple representing  a minimal absent word in $x$, where for some minimal absent word $w$ of length $|w| > 2$,  $w = x[i \dd j]\alpha$. Notice that this representation is unique.

\begin{algo}{AbsentAvoidedWords}{$$x$, $k$, $\rho$$}
\SET{\mathcal{A}}{\Call{MinimalAbsentWords}{x}}
\DOFOR{\mbox{each tuple  $<(i,j),\alpha> \in \mathcal{A}$ such that $k =j-i+2$} }
\SET{u_p}{\Call{Node}{i,j}}
\IF {\Call{IsImplicit}{u_p}}
\SET{(u,v)}{\Call{Edge}{u_p}}
\SET{f_p}{\mathcal{C}(v)}
\ELSE
\SET{f_p}{\mathcal{C}(u_p)}
\FI
\SET{u_i}{\Call{Node}{i+1,j}}
\IF {\Call{IsImplicit}{u_i}}
\SET{(u,v)}{\Call{Edge}{u_i}}
\ACT{f_i \leftarrow f_s \leftarrow \mathcal{C}(v)}
\ELSE
\SET{f_i}{\mathcal{C}(u_i)}
\SET{u_s}{\Call{Child}{u_i,\alpha}}	
 \SET{f_s}{\mathcal{C}(u_s)}	
\FI
\SET{E}{f_p\times f_s/f_i}
\IF{(0 - E)/(max\{1,\sqrt{E}\}) \leq \rho}
\CALL{Report}{x[i\dd j]\alpha}
\FI
\OD
\end{algo}

Intuitively, the idea is to check the length of every minimal absent word. If a tuple $< (i,j), \alpha  >$ represents a minimal absent word $w$ of length $k = j-i+2$, then the value of  $\std(w)$ is computed to determine whether $w$ is an absent $\rho$-avoided word. Note that, if $w = x[i\dd j]\alpha$ is a minimal absent word, then $w_p= x[i\dd j]$, $w_i= x[i+1 \dd j]$, and $w_s = x[i+1\dd j]\alpha$ occur in $x$ by Definition~\ref{def:maws}. Thus, there are three (implicit or explicit) nodes in $\mathcal{T}(x)$  path-labelled  $w_p$, $w_i$, and $w_s$, respectively. The observed frequencies of $w_p$, $w_i$, and $w_s$ are already computed during the preprocessing of $\mathcal{C}$, which stores the number of terminal nodes in the subtree rooted at $v$, for each node $v$. 

Notice that for an explicit node $v$ path-labelled $w'= x[i' \dd j']$, the value $\mathcal{C}(v)$ represents the number of occurrences (observed frequency) of $w'$ in $x$; whereas for an implicit node along the edge $(u,v)$  path-labelled  $w''$, then the number of occurrences of $w''$ is equal to  $\mathcal{C}(v)$ (and not $\mathcal{C}(u)$). The implementation of this procedure is given in Routine \Algo{AbsentAvoidedWords}.

\subsection{Computing Occurring Avoided  Words}

Lemma~\ref{lem:occ} suggests that for each occurring $\rho$-avoided word $w$, $w_p$ is a path-label of an  explicit node $v$ of $\mathcal{T}(x)$. Thus, for each internal node $v$ such that  $\mathcal{D}(v)= k-1$ and $\mathcal{L}(v)= w_p$, Routine  \Algo{OccurringAvoidedWords} computes $\std(w)$, where $w =w_p \alpha$ is a path-label of a child (explicit or implicit) node of $v$. Note that if $w_p$ is a path-label of an explicit node $v$ then $w_i$ is  a path-label of an explicit node $u$ of $\mathcal{T}(x)$; node $u$ is well-defined and it is  the node pointed at by the suffix-link of $v$.  The implementation of this procedure is given in Routine \Algo{OccurringAvoidedWords}.

\begin{algo}{OccurringAvoidedWords}{$$x$, $k$, $\rho$$}
\SET {N}{\mbox{~an empty stack}}
\CALL{Push}{N, root(x)}
\DOWHILE{ N \mbox{~is not empty}}
\SET{u}{\Call{Pop}{N}}
\IF {u \mbox{~is not labelled as discovered}}
\CALL{Label}{u}
\FI
\DOFOR{\mbox{each edge  $(u,v)$ of $\mathcal{T}(x)$}}
\IF{\mathcal{D}(v) < k-1}
\CALL{Push}{N,v}
\ELSEIF{\mathcal{D}(v) = k-1}
\SET{f_p}{\mathcal{C}(v)}
\SET{f_i}{\mathcal{C}(\mbox{\textit{suffix-link}}[v])}
\DOFOR{\mbox{each child  $v'$ of $v$}}
	\SET{f_w}{\mathcal{C}(v')}
    \SET{\alpha}{\mathcal{L}(v')[k]}
	\SET{f_s}{\mathcal{C}(\Call{Child}{\mbox{\textit{suffix-link}}[v],\alpha})}
	\SET{E}{f_p\times f_s/f_i}
	\IF{(f_w-E)/(max\{1,\sqrt{E}\}) \leq \rho} 
	\CALL{Report}{\mathcal{L}(v')[0\dd k-1]}\FI
\OD
\FI
\OD
\OD
\end{algo}

\subsection{Algorithm Analysis}

\begin{lemma} \label{lem:cor}
Given a word $x$, an integer $k>2$, and a real number  $\rho < 0$,  Algorithm  \Algo{AvoidedWords} computes all $\rho$-avoided words of length $k$ in $x$.
\end{lemma}

\begin{proof}
By definition, a $\rho$-avoided word $w$ is either an absent $\rho$-avoided word or an occurring one. Hence,
the proof of correctness relies on Lemma~\ref{lem1} and Lemma~\ref{lem:occ}. 
First,  Lemma~\ref{lem1} indicates that an absent $\rho$-avoided word in $x$ is necessarily a minimal absent word. Routine \Algo{AbsentAvoidedWords} considers each minimal absent word $w$ and verifies if $w$ is a $\rho$-avoided word of length $k$.

Second,  Lemma~\ref{lem:occ} indicates that  for each occurring $\rho$-avoided word $w$, $w_p$ is a path-label of an  explicit node $v$ of $\mathcal{T}(x)$. Routine \Algo{OccurringAvoidedWords} considers each child of such node of word-depth $k$, and verifies if its path-label is a $\rho$-avoided word. \qed
\end{proof}

\begin{lemma} 
\label{lem:ana}
Given a word $x$ of length $n$ over a fixed-sized alphabet, an integer $k>2$ and a real number  $\rho < 0$,  Algorithm  \Algo{AvoidedWords} requires time and space $\cO(n)$; for integer alphabets, it requires time $\cO(\sigma n)$. 
\end{lemma}
 \begin{proof}
Constructing the suffix tree $\mathcal{T}(x)$ of the input word $x$  takes time and space $\cO(n)$ for word over a fixed-sized alphabet~\cite{CHL07}. Once the suffix tree is constructed, computing arrays  $\mathcal{D}$ and $\mathcal{C}$ by traversing $\mathcal{T}(x)$ requires time and space $\cO(n)$. Note that the path-labels of the nodes of $\mathcal{T}(x)$ can by implemented in time and space $\cO(n)$ as follows: traverse the suffix tree to compute for each  node $v$ the smallest index $i$ of the terminal nodes of the subtree rooted at $v$. Then $\mathcal{L}(v) =  x[i\dd i+\mathcal{D}(v)-1]$.   

Next, Routine \Algo{AbsentAvoidedWords} requires time $\cO(n)$. It starts by computing all minimal absent words of $x$, which can be achieved in time and space $\cO(n)$ over a fixed-sized alphabet~\cite{Barton2014}. The rest of the procedure deals with checking each of the $\cO(n)$ minimal absent words of length $k$. Checking each minimal absent word $w$ to determine whether it is a $\rho$-avoided word or not requires time $\cO(1)$. In particular, an $\cO(n)$-time preprocessing of $\mathcal{T}(x)$ allows the retrieval of the (implicit or explicit) node in $\mathcal{T}(x)$ corresponding to  the longest proper prefix of $w$ in time $\cO(1)$~\cite{GawrychowskiLN14}. 
Finally, Routine \Algo{OccurringAvoidedWords} requires time $\cO(n)$. It traverses the suffix tree $\mathcal{T}(x)$ to allocate all  explicit node  of word-depth $k-1$. Then for each such node, the procedure check every (explicit or implicit) child  of word-depth $k$. The total number of these children is at most $n-k+1$. For every  child node, the procedure checks whether its  path-label  is a $\rho$-avoided word in time $\cO(1)$. 

For integer alphabets, the suffix tree can be constructed in time $\cO(n)$~\cite{farach1997optimal} and all minimal absent words can be computed in time $\cO(\sigma n)$~\cite{Barton2014}. The efficiency of Algorithm \Algo{AvoidedWords} is then limited by the total number of words to be considered, which, by Lemma~\ref{lem2}, is bounded from above by $\cO(\sigma n)$. \qed
 \end{proof}
\noindent Lemmas~\ref{lem:cor} and~\ref{lem:ana} imply the first result of this article.
 \begin{theorem}
Algorithm \Algo{AvoidedWords} solves Problem \textsc{AvoidedWordsComputation} in time and space $\cO(n)$. For integer alphabets, the algorithm solves the problem  in time $\cO(\sigma n)$. 
\end{theorem}
\subsection{Optimal Computation of all $\rho$-Avoided Words}
Although the biological motivation is yet to be shown for this, we show here how we can modify Algorithm  \Algo{AvoidedWords} so that it computes {\em all} $\rho$-avoided words (of all lengths) in a given word $x$ of length $n$ over an alphabet of size $\sigma$ in $\cO(\sigma n)$-time and $\cO(\sigma n)$-space. We further show that this algorithm is in fact time-optimal.
\begin{lemma}\label{lem7}
The upper bound $\cO(\sigma n)$ on the number of minimal absent words of a word 
of length $n$ over an alphabet of size $\sigma$ is tight if $2 \leq \sigma \leq n$.
\end{lemma}
\begin{proof}
Let $\Sigma =\{a_1,a_2\}$, i.e.~$\sigma=2$, and consider the word $x=a_2 a_1^{n-2} a_2$ of length $n$. All words of the form $a_2 a_1^k a_2$ for $0 \leq k \leq n-3$ are minimal absent words in $x$. Hence $x$ has at least $n-2=\Omega(n)$ minimal absent words.

Let $\Sigma=\{a_1,a_2,a_3,\ldots,a_\sigma\}$ with $3 \leq \sigma \leq n$, and consider the word $x=a_2 a_1^k a_3 a_1^k a_4 a_1^k \ldots a_i a_1^k a_{i+1} \ldots a_{\sigma} a_1^k a_1^{n-(\sigma-1)(k+1)}$, where $k=\lfloor \frac{n}{\sigma-1}\rfloor -1$. Note that $|x|=n$. Further note that $a_i a_1^j$ is a factor of $x$ for all $2 \leq i \leq \sigma $ and $0 \leq j \leq k$. Similarly, $a_1^j a_l$ is a factor of $x$ for all $3 \leq l \leq \sigma $ and $0 \leq j \leq k$. Thus all proper factors of all the strings in set $S=\{ a_i a_1^j a_l \: | \: 0 \leq j \leq k, \: 2 \leq i \leq \sigma, \: 3 \leq l \leq \sigma \}$ occur in $x$. The only strings in $S$ though that occur in $x$ are the ones of the form $a_i a_1^k a_{i+1}$, for all $2 \leq i < \sigma$.
Hence $x$ has at least $(\sigma-1)(\sigma-2)(k+1)-(\sigma-2)=(\sigma-1)(\sigma-2)\lfloor \frac{n}{\sigma-1}\rfloor-(\sigma-2)=\Omega(\sigma n)$ minimal absent words.
\qed 
\end{proof}
\begin{lemma}\label{lem8}
The total number of $\rho$-avoided words in a word $x$ of length $n$ over an alphabet of size $\sigma \leq n$ is bounded from above by $\cO(\sigma n)$ and this bound is tight.
\end{lemma}
\begin{proof}
By Lemma~\ref{lem1}, every $\rho$-avoided word is either occurring or a minimal absent word.
The set of occurring $\rho$-avoided words in $x$ can be injected to the set of explicit nodes of $\mathcal{T}(x)$ by Lemma~\ref{lem:occ}. It is well known that the number of explicit nodes of $\mathcal{T}(x)$ is $\Theta (n)$~\cite{CHL07} (at most $2n$) and hence it follows that the number of occurring $\rho$-avoided words is $\cO(n)$. Furthermore it is known that the total number of minimal absent words in $x$ is $\cO(\sigma n)$~\cite{Mignosi02}. Hence the number of $\rho$-avoided words is bounded from above by $\cO(\sigma n)$. Based on Lemma~\ref{lem7}, we know that for any alphabet of size $2 \leq \sigma \leq n$ there exist words with $\Omega(\sigma n)$ minimal absent words. Consider such a word and some $\rho \geq - \frac {1}{n}$. Then every minimal absent word is $\rho$-avoided since for any such word $E(w) \geq \frac {1}{n}$, $f(w)=0$ and hence $std(w) \leq - \frac {1}{n} \leq \rho$. Thus the bound is attainable. \qed 
\end{proof}
It is clear that if we just remove the condition on the length of each minimal absent word in Line 2 of \Algo{AbsentAvoidedWords} we then compute all absent $\rho$-avoided words in time and space $\cO(\sigma n)$. In order to compute all occurring $\rho$-avoided words in $x$ it suffices by Lemma~\ref{lem:occ} to  investigate the children of explicit nodes. We can thus traverse the suffix tree $\mathcal{T}(x)$ and for each explicit internal node, check for all of its children (explicit or implicit) whether their path-label is a $\rho$-avoided word. We can do this in $\cO(1)$ time as above. The total number of these children is at most $2n-1$, as this is the bound on the number of edges of $\mathcal{T}(x)$~\cite{CHL07}.
This modified algorithm is clearly time-optimal for fixed-sized alphabets as it then runs in time and space $\cO(n)$. The time optimality for integer alphabets follows directly from Lemmas~\ref{lem7} and~\ref{lem8}. Hence we obtain the following result.
\begin{theorem}
Given a word $x$ of length $n$ over an integer alphabet of size $\sigma$ and a real number $\rho < 0$, all $\rho$-avoided words in $x$ can be computed in time and space $\cO(\sigma n)$. This is time-optimal if $\sigma \leq n$.
\end{theorem}
\begin{lemma}
The expected length of the longest $\rho$-avoided word in a word $x$ of length $n$ over an alphabet of size $\sigma>1$ is $\cO(\log_{\sigma} n)$ when the letters are independent and identically distributed random variables uniformly distributed.
\end{lemma}
\begin{proof}
By Lemma~\ref{lem:occ} the length of the longest occurring word is bounded above by the word-depth of the deepest internal explicit node in $\mathcal{T}(x)$ incremented by 1. We note that the greatest word-depth of an internal node corresponds to the longest repeated factor in word $x$. Moreover, for a word $w$ to be a minimal absent word, $w_i$ must appear at least twice in $x$ (in the occurrences of $w_p$ and $w_s$). Hence the length of the longest $\rho$-avoided word is bounded by the length of the longest repeated factor in $x$ incremented by 2. The expected length of the longest repeated factor in a word is known to be $\cO(\log_{\sigma} n)$~\cite{SA} and hence the lemma follows.\qed
\end{proof}
\section{Implementation and Experimental Results}

Algorithm $\Algo{AvoidedWords}$ was implemented as a program to compute the $\rho$-avoided words of length $k$ in one or more input sequences. The program was implemented in the \texttt{C++} programming language and developed under GNU/Linux operating system. The input parameters are a (Multi)FASTA file with the input sequences(s), an integer $k > 2$, and a real number $\rho < 0$. The output is a file with the set of $\rho$-avoided words of length $k$ per input sequence. The implementation is distributed under the GNU General Public License, and it is available at \url{http://github.com/solonas13/aw}. The experiments were conducted on a Desktop PC using one core of Intel Core i5-4690 CPU at 3.50GHz under GNU/Linux. The programme was compiled with \texttt{g++} version 4.8.4 at optimisation level 3 (-O3). We also implemented a brute-force approach for the computation of $\rho$-avoided words. We mainly used it to confirm the correctness of our implementation. Here we do not plot the results of the brute-force approach as it is easily understood that it is orders of magnitude slower than our approach. 

\begin{figure}[t]
	\begin{center}
	\subfloat[Time for $n=1$\scriptsize{M} and $\rho=-10$]{\includegraphics[width=0.35\textwidth,angle=270]{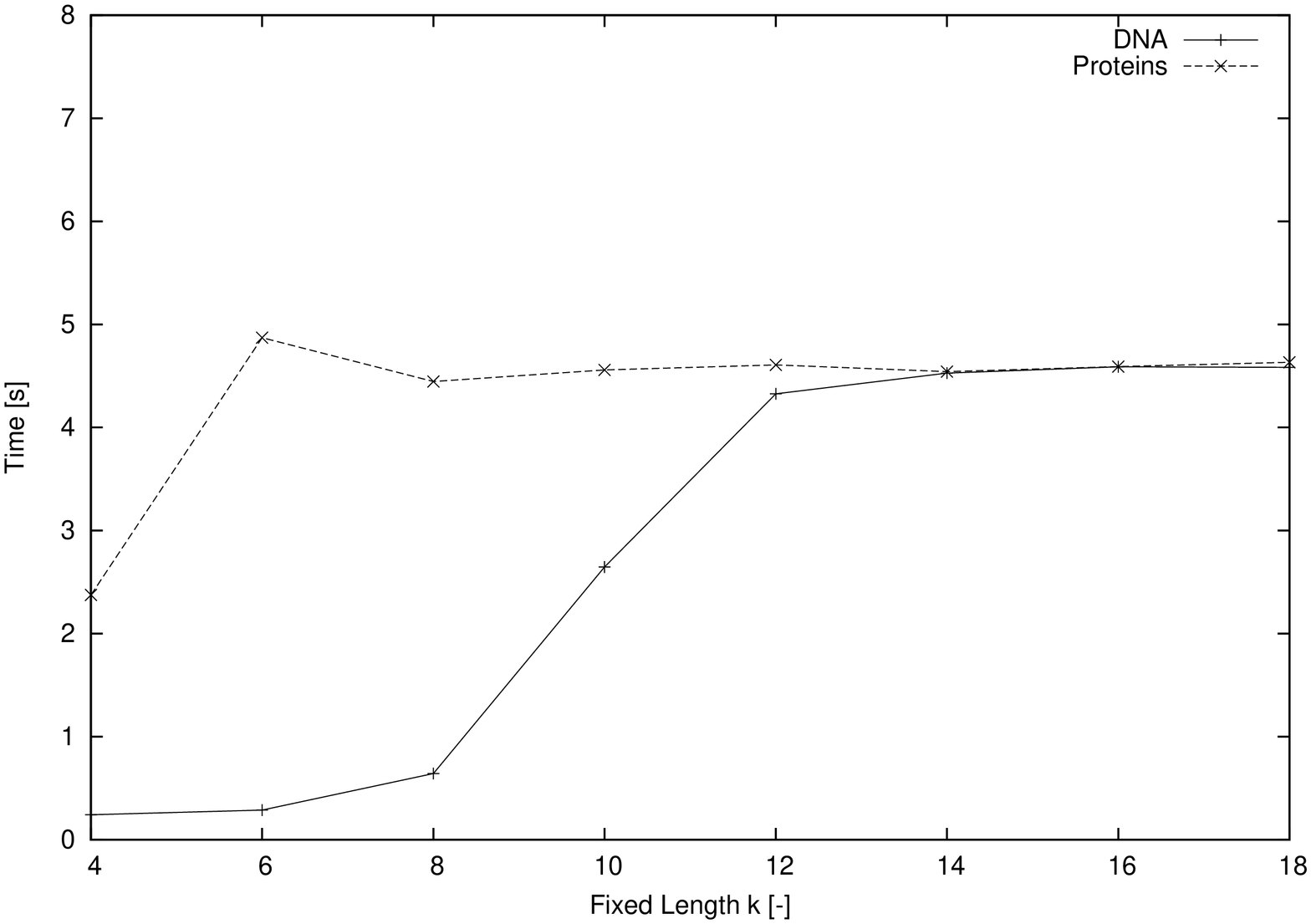}}
	\subfloat[Time for $n=1$\scriptsize{M} and $k=8$]{\includegraphics[width=0.35\textwidth,angle=270]{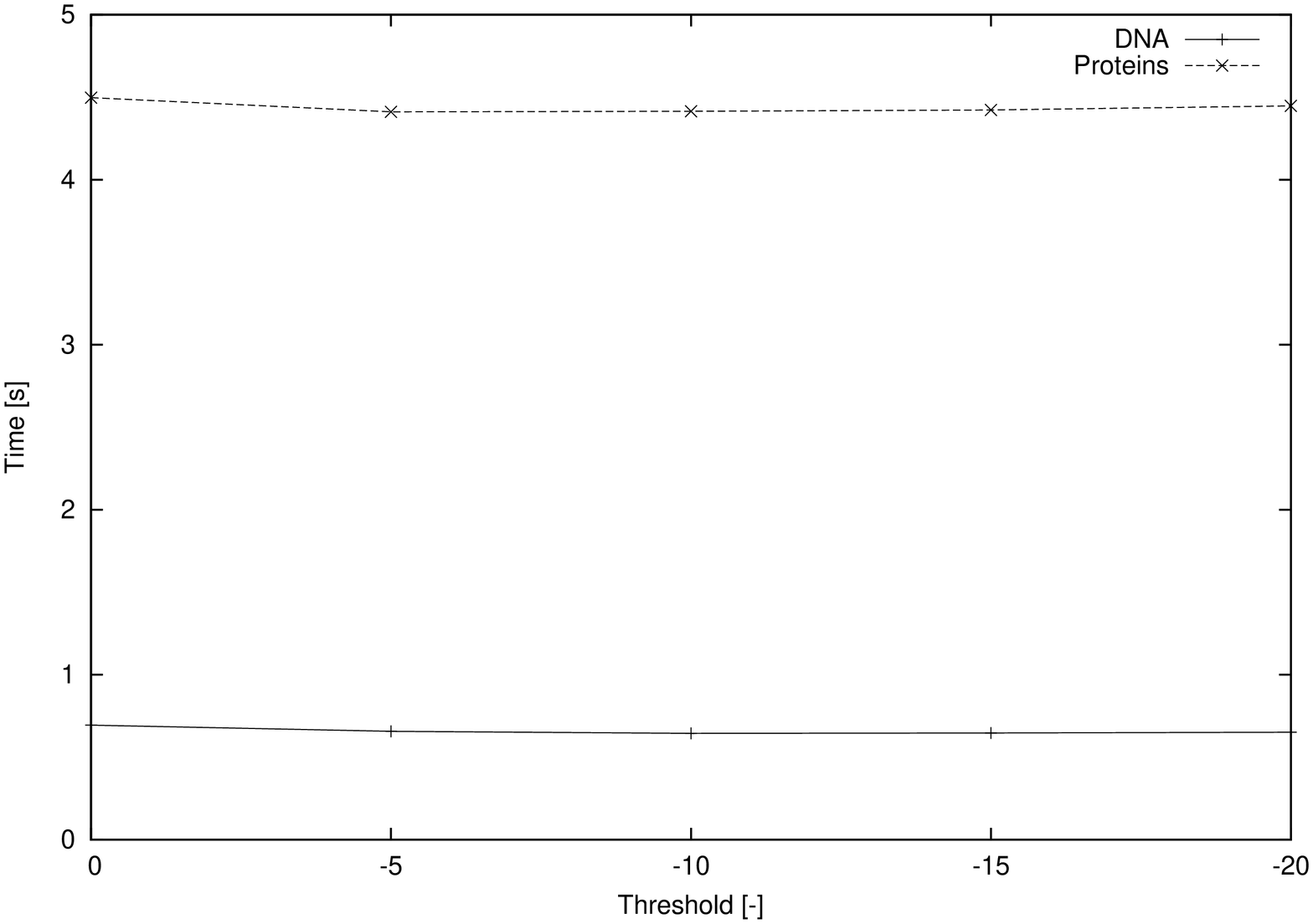}}
	\caption{Elapsed time of Algorithm $\Algo{AvoidedWords}$ using synthetic DNA ($\sigma=4$) and proteins ($\sigma=20$) data of length $1$M for variable $k$ and variable $\rho$.}
	\label{fig:fixed}
	\end{center}
\end{figure}

To evaluate the time performance of our implementation, synthetic DNA ($\sigma=4$) and proteins ($\sigma=20$) data were used. The input sequences were generated using a randomised script. In the first experiment, our task was to establish that the performance of the program does not essentially depend on $k$ and $\rho$; i.e., the elapsed time of the program remains unchanged up to some constant with increasing values of $k$ and decreasing values of $\rho$.  As input datasets, for this experiment, we used a DNA and a proteins sequence both of length $1$M (1 Million letters). For each sequence we used different values of $k$ and $\rho$. The results, for elapsed time are plotted in Fig.~\ref{fig:fixed}. It becomes evident from the results that the time performance of the program remains unchanged up to some constant. The longer time required for the proteins sequences for small values of $k$ is explained by the increased number of branching nodes in this depth in the corresponding suffix tree due to the size of the alphabet ($\sigma=20$). To confirm this we counted the number of nodes considered by the algorithm to compute the $\rho$-avoided words for $k=4$ and $\rho=-10$ for both sequences. The number of considered nodes for the DNA sequence was $260$ whereas for the proteins sequence it was $1,585,510$. 

In the second experiment, our task was to establish the fact that the elapsed time and memory usage of the program grow linearly with $n$, the length of the input sequence.  As input datasets, for this experiment, we used synthetic DNA and proteins sequences ranging from $1$ to $128$ M. For each sequence we used constant values for $k$ and $\rho$: $k=8$ and $\rho=-10$. The results, for elapsed time and peak memory usage, are plotted in Fig.~\ref{fig:linear}. It becomes evident from the results that the elapsed time and memory usage of the program grow linearly with $n$. The longer time required for the proteins sequences compared to the DNA sequences for increasing $n$ is explained by the increased number of branching nodes in this depth ($k=8$) in the corresponding suffix tree due to the size of the alphabet ($\sigma=20$).  To confirm this we counted the number of nodes considered by the algorithm to compute the $\rho$-avoided words for $n=64$M for both the DNA and the proteins sequence. The number of nodes for the DNA sequence was $69,392$ whereas for the proteins sequence it was $43,423,082$.

\begin{figure}[!t]
	\begin{center}
	\subfloat[Time for $k=8$ and $\rho=-10$] {\includegraphics[width=0.35\textwidth,angle=270]{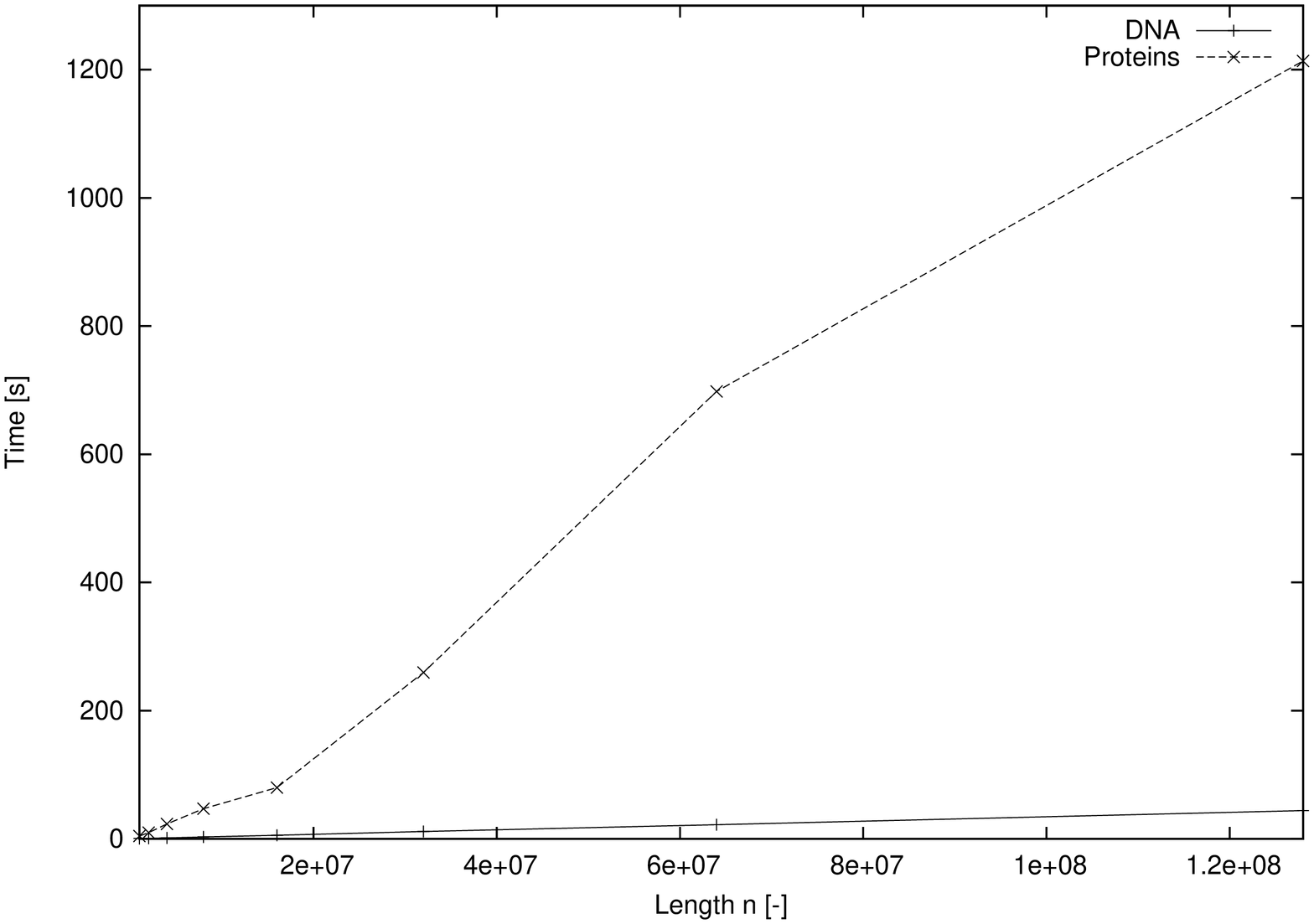}}
	\subfloat[Memory for $k=8$ and $\rho=-10$]{\includegraphics[width=0.35\textwidth,angle=270]{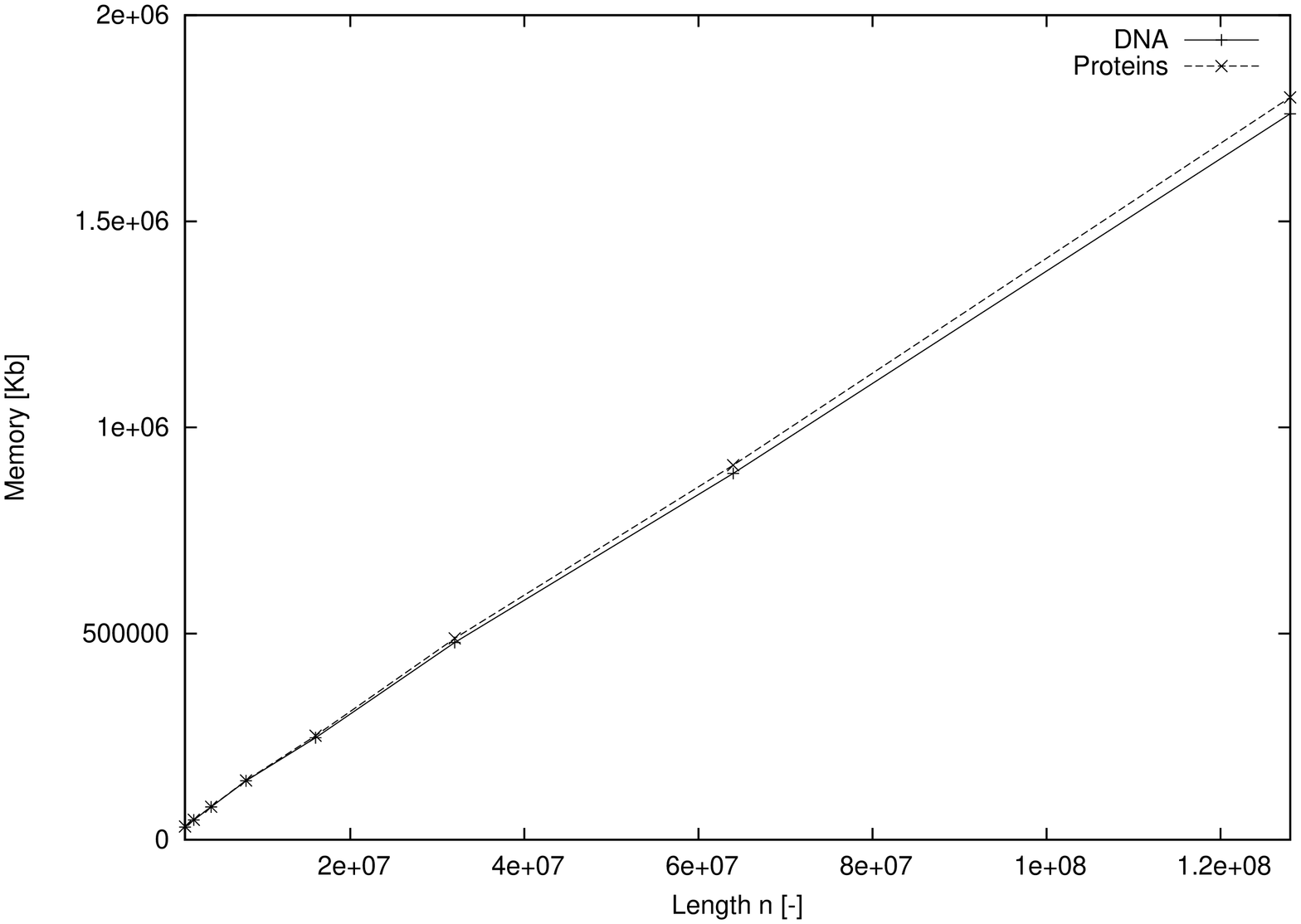}}
	\caption{Elapsed time and peak memory usage of Algorithm $\Algo{AvoidedWords}$ using synthetic DNA ($\sigma=4$) and proteins ($\sigma=20$) data of length $1$M to $128$M.}
	\label{fig:linear}
	\end{center}
\end{figure}

In the next experiment, our task was to evaluate the time and memory performance of our implementation with real data. As input datasets, for this experiment, we used all chromosomes of the human genome. Their lengths range from around $46$M (chromosome 21) to around $249$M (chromosome 1). For each sequence we used $k=8$ and $\rho=-10$. The results, for elapsed time and peak memory usage, are plotted in Fig.~\ref{fig:real}. The results with real data confirm that the elapsed time and memory usage of the program grow linearly with $n$.

\begin{figure}[!t]
	\centering
	\subfloat[Time for $k=8$ and $\rho=-10$]{\includegraphics[width=0.35\textwidth,angle=270]{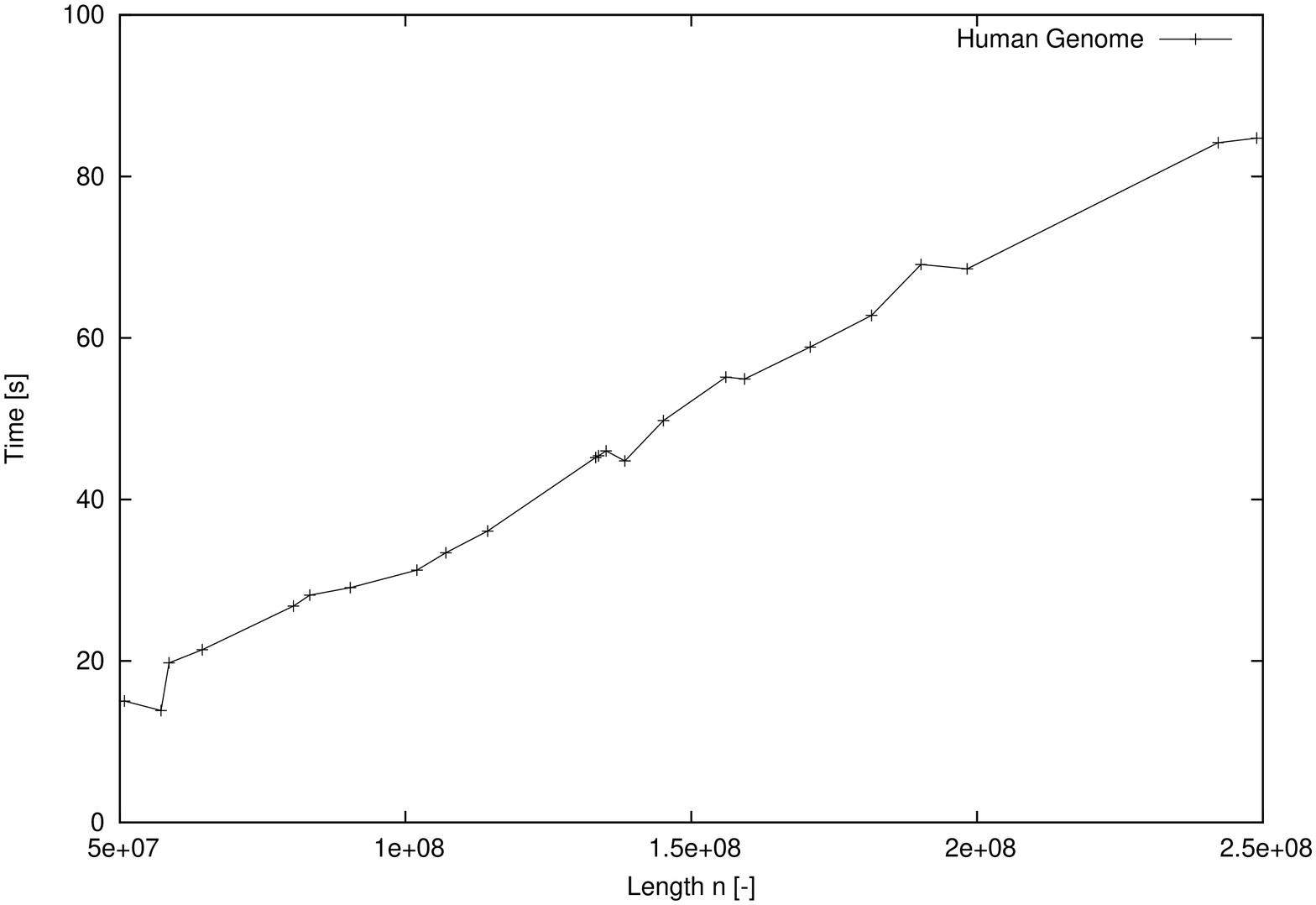}}
	\subfloat[Memory for $k=8$ and $\rho=-10$]{\includegraphics[width=0.35\textwidth,angle=270]{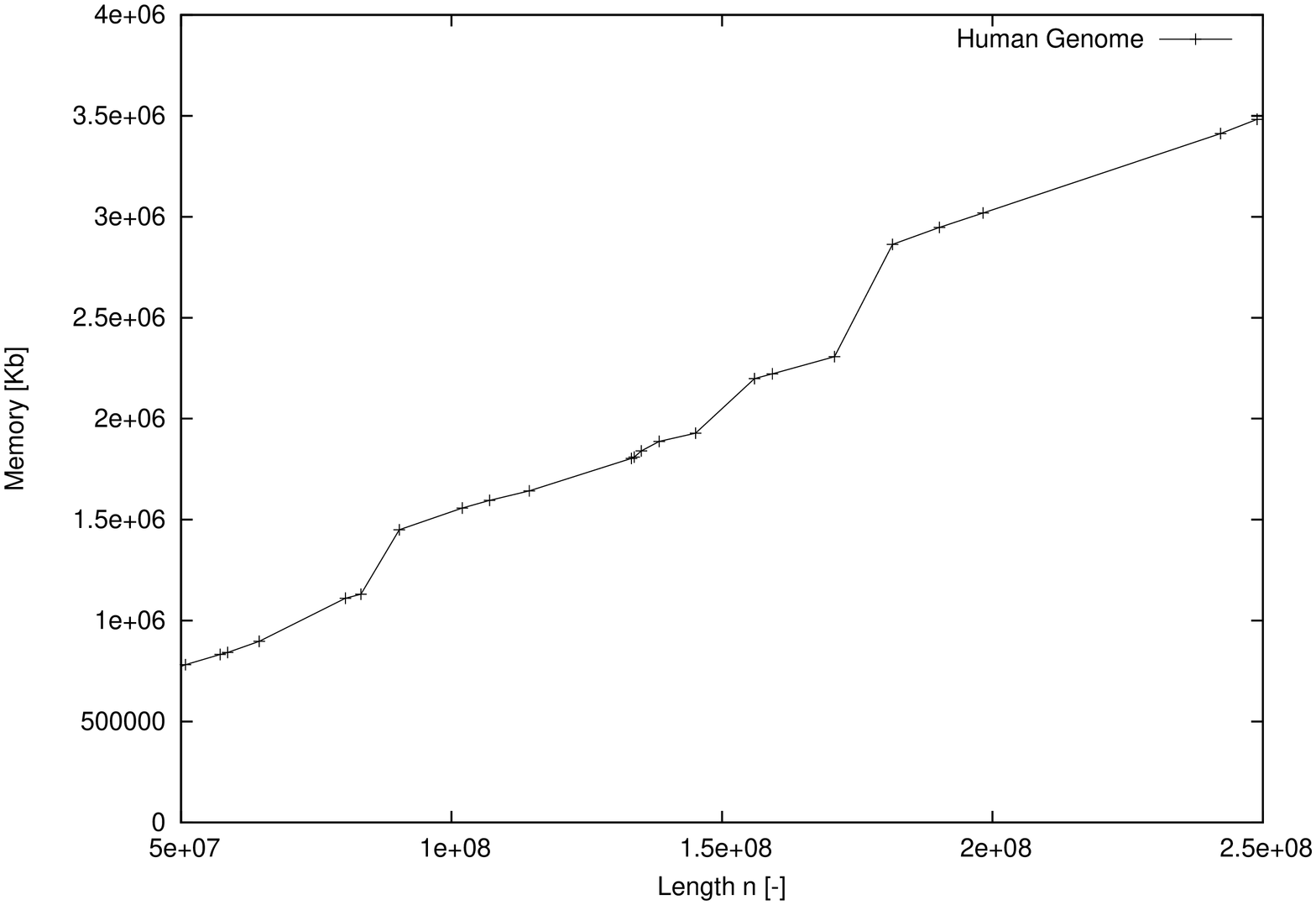}}
	\caption{Elapsed time and peak memory usage of Algorithm $\Algo{AvoidedWords}$ using all chromosomes of the human genome.}
	\label{fig:real}
\end{figure}

\paragraph*{Real Application.} We computed the set of avoided words for $k=6$ (hexamers) and $\rho=-10$
in the complete genome of {\em E. coli} and sorted the output in increasing order of their standard deviation. The most avoided words were extremely enriched in {self-complementary} (palindromic) hexamers. In particular, within the output of 28 avoided words, 23 were self-complementary; and the 17 most avoided ones were {\em all} self-complementary. For comparison, we computed the set of avoided words for $k=6$ and $\rho=-10$ from an eukaryotic sequence: a segment of the human chromosome 21 (its leftmost segment devoid of \texttt{N}'s) equal to the length of the {\em E. coli} genome. In the output of 10 avoided words, no self-complementary hexamer was found. Our results confirm that the restriction endonucleases which target self-complementary sites are not found in eukaryotic sequences~\cite{Rusinov2015}.

Our immediate target is to investigate the avoidance of words in Genomic Regulatory Blocks (GRBs)~\cite{Akalin2009} within the same organism and across evolution.

%


\bibliographystyle{splncs03}
\bibliography{reference}

\end{document}